\documentclass[11pt,draftcls,onecolumn]{IEEEtran} 
\usepackage{fullpage,amsmath,amssymb,amsthm,citesort}
\usepackage{epsfig,graphicx,subfigure,wrapfig,psfrag}		
\graphicspath{{./figs}}

\def\manifold{\mathcal{M}}
\def\reals{\mathbb{R}}
\def\comps{\mathbb{C}}
\def\sphere{\mathbb{S}}
\def\vol{\mbox{vol}}
\def\diag{\mbox{diag}}

\newcommand{\Tan}[1]{\mathcal{T}_{#1}\manifold}

\newcommand{\cover}[3]{\mathcal{C}\left(#1,\; #2, \; #3 \right)}
\newtheorem{definition}{Definition}[section]
\newtheorem{lemma}{Lemma}[section]
\newtheorem{thm}{Theorem}[section]

\newtheorem{cor}{Corollary}[section]

\begin{document}

\title{Stable Manifold Embeddings with\\ Structured Random Matrices}
\date{}
\author{Han Lun Yap, Michael B. Wakin, and Christopher J. Rozell\thanks{
Copyright (c) 2013 IEEE. Personal use of this material is permitted. However, permission to use this material for any other purposes must be obtained from the IEEE by sending a request to pubs-permissions@ieee.org.

HLY and CJR are with the School of Electrical and Computer Engineering at the Georgia Institute of Technology. MBW is with the Department of Electrical Engineering and Computer Science at the Colorado School of Mines. This work was partially supported by NSF grants CCF-0830456 and CCF-0830320, by NSF CAREER grant CCF-1149225, and by DSO National Laboratories of Singapore. The authors are grateful to A. Eftekhari for valuable discussions about this work. A preliminary version of this work appeared in the Proceedings of the 45th Annual Conference on Information Sciences and Systems~\cite{Yap2011d}.}}

\maketitle

\begin{abstract}
The fields of compressed sensing (CS) and matrix completion have shown that high-dimensional signals with sparse or low-rank structure can be effectively projected into a low-dimensional space (for efficient acquisition or processing) when the projection operator achieves a stable embedding of the data by satisfying the Restricted Isometry Property (RIP). It has also been shown that such stable embeddings can be achieved for general Riemannian submanifolds when random orthoprojectors are used for dimensionality reduction. Due to computational costs and system constraints, the CS community has recently explored the RIP for structured random matrices (e.g., random convolutions, localized measurements, deterministic constructions).  The main contribution of this paper is to show that any matrix satisfying the RIP (i.e., providing a stable embedding for sparse signals) can be used to construct a stable embedding for manifold-modeled signals by randomizing the column signs and paying reasonable additional factors in the number of measurements, thereby generalizing previous stable manifold embedding results beyond unstructured random matrices.   We demonstrate this result with several new constructions for stable manifold embeddings using structured matrices. This result allows advances in efficient projection schemes for sparse signals to be immediately applied to manifold signal models.

\end{abstract}

\section{Introduction}
\label{sec:intro}

Much of modern signal processing rests on the observation that many high-dimensional signals of interest in fact have an intrinsic low-dimensional structure that can be described with a geometric model.
For example, sparse signals live on a union of low-dimensional subspaces within an ambient high-dimensional signal space~\cite{Baraniuk2008a}, while parametric signals and certain non-parametric signal collections are constrained to live on (or near) low-dimensional manifolds~\cite{DonohoGrimesISOMAP,Eigenfaces}. While this low-dimensional structure can be exploited to great effect in signal processing applications, the high-dimensionality of the ambient signal space can severely complicate the acquisition and processing of the data~\cite{Baraniuk2011}. 
To partially address this issue, several recent results have shown that compressive linear operators can provide stable embeddings that preserve the geometry of the signal model (i.e., preserve pairwise points between signals) in a lower-dimensional space.

Much of the work on compressive linear operators has come in the field of compressed sensing (CS), where it is known that certain randomized compressive matrix constructions will satisfy a condition known as the Restricted Isometry Property (RIP)~\cite{Candes2006} with high probability.
The RIP guarantees that a matrix will approximately preserve distances between all pairs of sparse signals, therefore stably embedding the signal model by preserving the geometric structure of the union of subspaces in the compressed (i.e., measurement) space. The RIP is a sufficient condition to guarantee robust recovery of sparse signals from their measurements via solving a computationally tractable $\ell_1$-minimization program.
In a similar vein, an equivalent formulation of the RIP for preserving distances between low-rank matrices also leads to matrix recovery guarantees from underdetermined linear measurements~\cite{Recht2007}.

The notion of a stable embedding, as quantified in the RIP, has also been extended to describe linear operators acting on signals living on a low-dimensional manifold~\cite{baraniuk2009random,Clarkson2008a}.
For example, it has been shown that an undersampled random orthoprojector can stably embed a manifold from a high-dimensional space into a lower-dimensional space~\cite{baraniuk2009random,Clarkson2008a}. 
Such stable embeddings are valuable because they ensure that key properties of the manifold are retained in the low-dimensional measurement space where processing is much more computationally efficient.
In particular, a stable embedding is a sufficient condition for guarantees on our ability to recover the original signal via tractable recovery algorithms~\cite{Shah2011} and for performance guarantees on data processing or inference algorithms in the measurement space~\cite{Davenport2010a}.
Moreover, a stable embedding also guarantees that manifold learning algorithms (e.g., Isomap~\cite{ISOMAP}) can be applied in the low-dimensional measurement space nearly as accurately as in the original signal space~\cite{hegde2007random}.

Recently, the CS community has turned to investigating structured measurement systems because unstructured systems (i.e., those corresponding to i.i.d.\ random matrices or random orthoprojectors that are classically analyzed in the CS literature) may be impractical due to memory constraints, computational costs, or limitations in the data acquisition architecture.
Several structured CS systems (e.g., random convolution systems described by partial Toeplitz~\cite{Haupt2010} and circulant matrices~\cite{Rauhut2010c,Krahmer2012}, localized sensing systems described by randomized block diagonal matrices~\cite{Yap2011a}, and certain deterministic matrix constructions~\cite{Devore2007}) have been shown to satisfy the RIP while requiring (at least analytically) a small increase in the number of measurements beyond what is needed for an unstructured random matrix. While concerns about the practicality of unstructured measurements also apply to systems acquiring manifold-modeled signals, the existing stable embedding results for structured matrices apply only to sparse signal models.

The main contribution of this paper is to demonstrate that \emph{any} matrix satisfying the RIP for sparse signals (including structured measurement systems as described above) can be used to generate a stable embedding of a manifold by randomizing the column signs of the matrix.  Our main theorem statement gives an explicit recipe for using the RIP guarantee of a matrix to determine the number of measurements sufficient to guarantee (with a prescribed probability) a stable manifold embedding of a specified conditioning.
Thus, the main goal of this paper is to generalize the existing stable manifold embedding results~\cite{baraniuk2009random,Clarkson2008a} for unstructured matrices by paying a reasonable penalty in the number of measurements to accommodate any matrix for which the RIP is established.  As practical examples, we compute the number of measurements sufficient for stable manifold embeddings when using measurement systems constructed from several structured matrices studied in the CS literature, including subsampled Fourier transforms, random convolution matrices, block diagonal matrices, and certain deterministic matrices.  
We show that for many structured matrices of interest, it suffices to have a number of measurements that scales linearly with the dimension of the manifold and logarithmically with properties of the manifold (to be described in detail in Section~\ref{subsec:manifold_signals}) and the ambient dimension. 
Our work rests on a recent result~\cite{Krahmer2010a} showing that when the columns of an RIP matrix are modulated by a random sign sequence, the matrix will obey a form of the Johnson-Lindenstrauss (JL) lemma~\cite{Johnson84extensionsof} and can therefore provide a stable embedding of an arbitrary finite point cloud. 
Following similar arguments to~\cite{baraniuk2009random}, we extend the finite JL result to all points living on a manifold.

\section{Background}
\label{sec:bg}

\subsection{Stable Embeddings}

When $M < N$, a compressive linear operator $\Phi \in \reals^{M \times N}$
possesses a nullspace of dimension at least $N - M$.
Therefore, distinct signals may be mapped onto, or close to, the same measurement by the operator if their difference falls on or near its nullspace.  In any application with finite resolution or noise, instability can result if very different signals are mapped close together.  It is therefore critical that the geometry of the subset $\manifold \subset \reals^N$ of signals of interest be maintained in the measurement space $\reals^M$.
This geometry preservation idea forms the basis for the following definition of a stable embedding by an operator:
\begin{definition}
	\label{def:stable_emb_def}
	A linear operator ${\Phi}$ provides a \emph{stable embedding} of a subset $\manifold \subset \reals^N$ with conditioning $\delta_\manifold$ if for all pairs $x_1, x_2 \in \manifold$, we have
	\begin{eqnarray}
		\label{eq:SME_def1}
		(1-\delta_\manifold) \le \frac{\|{\Phi} x_1  - {\Phi} x_2 \|_2^2}{\|x_1 - x_2\|_2^2} \le (1+\delta_\manifold).
	\end{eqnarray}
\end{definition}

For a finite data cloud $\manifold$ (i.e., $|\manifold| < \infty$), a stable embedding is established by the \emph{Johnson-Lindenstrauss} (JL) lemma~\cite{Johnson84extensionsof}.
For many random operators $\Phi$~\cite{Krahmer2010a,dasgupta2002elementary}, the JL lemma states that for a stable embedding of the set $\manifold$ to hold with high probability, the number of measurements need only  scale with $\log(|\manifold|)$ and not with the size of the ambient signal space. 
In contrast, in CS the set $\manifold$ is comprised of all $S$-sparse vectors, $\manifold := \{x \in \reals^N \;|\; \|x\|_0 \le S \}$, where $\|x\|_0$ counts the number of non-zero entries in $x$.  For this signal family, the notion of a stable embedding is given by the RIP, defined as:
\begin{definition}
	\label{def:RIP_def}
	A linear operator $\Phi \in \reals^{M \times N}$ satisfies the \emph{Restricted Isometry Property} of \emph{order} $S$ and \emph{conditioning} $\delta$ (or RIP-($S,\delta$) in short) if for all $x \in \reals^N$ with at most $S$ non-zero entries, we have
	\begin{eqnarray*}
		(1-\delta)\|x\|_2^2 \le \|\Phi x\|_2^2 \le (1+\delta)\|x\|_2^2.
	\end{eqnarray*}
\end{definition}
Because the difference between $S$-sparse signals is at most $2S$-sparse, an operator satisfying RIP-($2S, \delta$) provides a stable embedding with conditioning $\delta$ of the union of all $S$-sparse subspaces of $\reals^N$.

\subsection{Manifold-modeled Signals}
\label{subsec:manifold_signals}

The sparsity and low-rank signal models that have gained significant attention in the signal processing community do not apply well to all signal families.
Instead, many high-dimensional signals can be modeled as lying on (or near) low-dimensional submanifolds embedded in Euclidean space.
One example class of such signals are parametric signals that are determined by a parameter $\theta \in \Theta$, where $\Theta$ is a $D$-dimensional (typically $D \ll N$) parameter space (which could be a $D$-dimensional manifold itself or simply a subset of $\reals^D$).
More precisely, a parametric signal class can be written as $\manifold := \{ x \in \reals^N \;|\; x = f(\theta), \theta \in \Theta \}$ where $f : \Theta \rightarrow \reals^N$ is a smooth function.\footnote{For $\manifold$ to be a Riemannian submanifold as required by our main result, additional conditions on the function $f$ may be necessary (e.g., $f$ should be a diffeomorphism).}
Examples of such parametric signals include a 1-dimensional signal parameterized by a time delay ($D=1$), a radar chirp characterized by its starting and ending time and frequency ($D=4$), and images of an object articulated in space~\cite{DonohoGrimesISOMAP}.
Not all manifold-modeled signals of interest can be parametrized. 
Nonetheless, low-dimensional submanifolds have also been useful as approximate models for nonparametric signal classes such as images of human faces~\cite{Eigenfaces} or hand-written digits~\cite{Hinton1997}.
We refer the reader to~\cite{Wakin} for further examples of interesting signal families that are well-modeled by low-dimensional submanifolds of Euclidean space.

Before discussing the stable embedding of manifolds, we establish some necessary notation and terminology.  In the remainder of this paper, we consider
$\manifold$ to be a Riemannian submanifold that inherits the canonical Euclidean metric from the ambient space.
For a given point $x$ on $\manifold$ embedded in $\reals^N$, we let $\Tan{x}$ denote the tangent space of $\manifold$ at $x$. 
As we are considering submanifolds $\manifold$ of dimension $D$ embedded in $\reals^N$, $\Tan{x}$ can be defined as a $D$-dimensional linear subspace of $\reals^N$ passing through the origin. 
We let $d_\manifold(x,y)$ denote the geodesic distance between two points $x,y \in \manifold$  (i.e., the length of the shortest path between $x$ and $y$ along the submanifold).

In this work, we consider two additional characterizations of a manifold that will be useful for describing certain local and global properties of the manifold.  
The first is the \emph{condition number} which provides a bound on the worst case curvature of any unit speed geodesic path along the manifold and a guarantee that the manifold is ``self-avoiding'' in that it does not curve back on itself at long geodesic distances.
The condition number, described in \cite{Niyogi2009} and used in~\cite{baraniuk2009random}, is typically denoted by the fraction $\frac{1}{\tau}$. 
Appendix~\ref{subsec:consequence_cond_num} describes in greater detail the implications of the condition number that are important in our proofs.  

A second useful quantity is the \emph{geodesic regularity} of a manifold.
Let $\vol(B)$ denote the volume of a set $B$, defined as $\vol(B) = \int_B dv$ where $dv$ is the volume element on $B$.
Next, for a Riemannian manifold $\manifold$, denote $B_\manifold(x, \epsilon)$ as the geodesic ball centered at $x \in \manifold$ of radius $\epsilon$, $B_\manifold(x, \epsilon) := \{ p \in \manifold \;|\; d_\manifold(p,x) \le \epsilon\}$.
Similarly, let $B_{\reals^D}(x, \epsilon)$ be the Euclidean ball of radius $\epsilon$ centered at $x \in \reals^D$, $B_{\reals^D}(x, \epsilon) := \{ p \in \reals^D \;|\; \|p - x\|_2 \le \epsilon\}$.
Then, the geodesic regularity $R$ is defined as follows:
\begin{definition}
	\label{def:geo_reg}
	A $D$-dimensional Riemannian submanifold $\manifold$ of $\reals^N$ has \emph{geodesic regularity} $R$ at resolution $\epsilon_0$ if for every $\epsilon \le \epsilon_0$ and for every $x \in \manifold$,
	\begin{eqnarray*}
		 \vol(B_{\reals^D}(0,\epsilon)) \le R^D \vol(B_\manifold(x,\epsilon)).
	\end{eqnarray*} 
\end{definition} 
We see that the geodesic regularity allows a uniform comparison of the geodesic balls and Euclidean balls (on the tangent spaces) of the same radius everywhere on the manifold.
This comparison is related to a certain intrinsic curvature (in particular, the scalar curvature) of the manifold~\cite{Gray1974}.
The consequences of the geodesic regularity $R$ on the covering numbers of a manifold (to be described later) are described in Appendix~\ref{sec:covering number}.
As in \cite{baraniuk2009random}, we shall subsequently neglect the minor dependence of the geodesic regularity $R$ on the maximum resolution $\epsilon_0$. 

As an illustration, we briefly describe a simple example of a parametric manifold and discuss its critical properties. Consider a sampled sinusoid given by
\begin{eqnarray*}
	s(\omega) = \left[e^{j\omega}, e^{j2\omega}, \cdots, e^{jN\omega} \right]^T \in \comps^N
\end{eqnarray*}
for $\omega \in [0, 2\pi)$. Observe that $\manifold := \{s(\omega) \;|\; \omega \in [0, 2\pi)\}$ is a $D=1$ dimensional submanifold in $\comps^N$ (which is isometric to $\reals^{2N}$). Lemma~\ref{lem:cond_properties} outlines the three implications of the condition number that are critical to our main results in this paper. In~\cite{Yap2013}, we show that the manifold $\manifold$ of sampled sinusoids satisfies the three properties of Lemma~\ref{lem:cond_properties} for\footnote{Notice that we are not deriving the condition number of this manifold. It suffices that the manifold satisfies the three properties of Lemma~\ref{lem:cond_properties} in order for Theorem~\ref{thm:main} to hold.}
\begin{eqnarray*}
	\frac{1}{\tau} := \frac{\sqrt{\sum_{n=1}^N n^4}}{\sum_{n=1}^N n^2},
\end{eqnarray*}
which scales like $N^{-0.5}$ for large $N$. Moreover, we also show that for this manifold, the volume 
$$
V = 2\pi \sqrt{ \sum_{n=1}^N n^2 }
$$
grows as $N^{1.5}$, and the geodesic regularity $R = 1$.

\subsection{Related Work}

The work in this paper is closely related to~\cite{baraniuk2009random} and~\cite{Clarkson2008a}, which both showed that with high probability, a random orthogonal projection $\Phi \in \reals^{M \times N}$ will provide a stable embedding of a $D$-dimensional submanifold $\manifold \subset \reals^N$ whenever $M$ scales linearly in $D$ and logarithmically in certain other parameters of the manifold.
We note that the main differences between these two works are that in~\cite{baraniuk2009random}, there is an additional dependence of $M$ on $\log(N)$, and that the manifold characterizations in both papers are slightly different.
In the present paper, we adopt the manifold characterizations presented in~\cite{baraniuk2009random}.
The proof of each of these results requires a finite covering of points carefully chosen from the manifold and a covering of the tangent planes of those points. 
Using the JL lemma previously described, it then can be argued that, with high probability, a random orthogonal projection will provide a stable embedding of these points.
Then, various geometric arguments allow one to conclude that the same orthogonal projection will provide a (slightly weaker) stable embedding of the entire manifold $\manifold$.

In this work, we adopt the same general proof approach but replace the JL lemma for random orthoprojectors with a JL lemma for operators satisfying the RIP. 
The following theorem, adapted from~\cite{Krahmer2010a}, expresses this JL lemma:
\begin{thm}
	\label{thm:Ward}
	Fix $0 < \rho, \epsilon < 1$ and suppose there is a finite set of points $E \subset \reals^N$.
	Also suppose we have a matrix $\Phi \in \reals^{M \times N}$ satisfying the RIP of order $S \ge 40 \log\left( \frac{4|E|}{\rho} \right)$ and conditioning $\delta \le \frac{\epsilon}{4}$.
	Let $\xi \in \reals^N$ be a Rademacher sequence (i.e., a sequence of i.i.d.\ equiprobable $\pm 1$ Bernoulli random variables), construct the diagonal Rademacher matrix $D_\xi := \diag(\xi)$, and define $\widehat{\Phi} := \Phi D_\xi \Psi$ where $\Psi \in \comps^{N \times N}$ is any unitary matrix.
	Then with probability exceeding $1 - \rho$,	we have for all $x \in E$, $(1 - \epsilon)\|x\|_2^2 \le \|\widehat{\Phi} x\|_2^2 \le (1 + \epsilon) \|x\|_2^2$.  
\end{thm}
In words, any operator satisfying the RIP can be used to approximately preserve the norms of any orthogonal transform of the signals in a given finite point cloud when the signs of the columns of the operator are randomly chosen.
We remark that if the finite point cloud $E$ is the set of all differences between points in another finite set $\manifold \subset \reals^N$, then a matrix $\Phi \in \reals^{M \times N}$ satisfying the RIP of order $S \ge 40 \log\left( \frac{4|\manifold|^2}{\rho} \right)$ (and conditioning $\delta \le \frac{\epsilon}{4}$) in Theorem~\ref{thm:Ward} can provide a stable embedding of $\manifold$ with high probability when the column signs of $\Phi$ are randomized.

\section{Stable Manifold Embeddings}

Section~\ref{sec:manRIP} contains a statement of our main result, showing that any matrix that satisfies the RIP (i.e., provides a stable embedding for sparse signals) can be used to form a stable embedding of a manifold.  Section~\ref{sec:manStruct} illustrates how this fact can be used to form stable manifold embeddings from several structured matrices that have been shown to satisfy the RIP.

\subsection{Manifold Embeddings from RIP Operators}
\label{sec:manRIP} 

Our main contribution, showing that RIP operators can be used to form stable manifold embeddings, is captured in the following theorem:
\begin{thm}
	\label{thm:main}
Let $\manifold$ be a compact $D$-dimensional Riemannian submanifold of $\reals^N$ with geodesic regularity $R$, volume $V$, and condition number $\frac{1}{\tau}$.
Suppose $\Phi \in \reals^{M \times N}$ is a matrix that satisfies RIP-($S,\delta$), and let $D_\xi \in \reals^{N \times N}$ be a diagonal Rademacher matrix.
Denote $\widehat{\Phi} = \Phi D_\xi \Psi$, where $\Psi \in \comps^{N \times N}$ is any unitary matrix.
Choose any conditioning $\delta_\manifold < 1$ and failure probability $\rho$.
If the RIP conditioning satisfies $\delta \le \frac{\delta_\manifold}{42}$ and the order $S$ of the RIP  satisfies
\begin{eqnarray*}
	S \ge 40 \left( 2D \log \left(\frac{3528 R \left(\sqrt{D/2 + 1} \right) (N+1)^2}{\sqrt{\pi} {\delta_\manifold^2} \tau} \right)  + (2D+1)\log\left(1 + \frac{21(N+1)}{\delta_\manifold} \right) + \log \left(\frac{8 V^2}{\rho}\right) \right), 
\end{eqnarray*}
then with probability exceeding $1 - \rho$, $\widehat{\Phi}$ provides a stable embedding of $\manifold$ with conditioning $\delta_{\manifold}$. 
\end{thm}
The proof of this theorem can be found in Appendix~\ref{sec:proof_main}.  Note that the theorem statement gives a clear recipe for both creating a stable manifold embedding from an RIP operator as well as determining how many measurements are sufficient to guarantee the desired result.  The main theorem statement relates the manifold properties to the required RIP order $S$, which can be related to the number of measurements by the original RIP proof for the operator in question (see also Section~\ref{sec:manStruct}).  We note especially that the RIP order only scales linearly with the manifold dimension $D$ and logarithmically with the ambient dimension $N$.  This is especially important because most interesting RIP results also have a linear relationship between the RIP order and the number of measurements.  Consequently, for such RIP results, this theorem allows the creation of a manifold embedding when the number of measurements scales linearly with the manifold dimension.  Once an RIP operator is generated with a sufficient number of measurements, a manifold embedding can be created by simply randomizing the column signs of the operator.

Sometimes, such as in manifold learning algorithms (e.g., Isomap~\cite{ISOMAP}), the main interest is in preserving the intrinsic geodesic distances between points of a data set lying on a submanifold of $\reals^N$ instead of their extrinsic Euclidean distances.   Prior work~\cite{baraniuk2009random} has shown that operators that stably embed a manifold with respect to Euclidean distances are also stable embeddings with respect to geodesic distances.
Therefore, stable embedding operators constructed according to Theorem~\ref{thm:main} also provide geodesic stable embeddings, guaranteeing that manifold learning algorithms can be performed significantly faster in the compressed space without much degradation~\cite{hegde2007random}.

\subsection{Manifold Embeddings from Structured Matrices}
\label{sec:manStruct}

As mentioned above, Theorem~\ref{thm:main} allows
us to construct operators providing stable manifold embeddings from any operator that satisfies the RIP.  We illustrate the implications of our result with a few notable examples below that establish stable manifold embeddings for operators with more structure than existing results on random orthoprojectors~\cite{baraniuk2009random}.  In the corollaries that follow, we assume that $\manifold$ is a compact $D$-dimensional Riemannian submanifold of $\reals^N$ with condition number $\frac{1}{\tau}$, volume $V$, and  geodesic regularity $R$.
We also assume a fixed failure probability $0 < \rho < 1$ and conditioning $0 < \delta_\manifold < 1$.
In what follows, we denote by $C_1, C_2, \cdots$ universal constants that do not depend on the other variables in the corollaries and that differ from corollary to corollary.

To begin, we consider a generalization of Gaussian random matrices to subgaussian random matrices (including Bernoulli, etc.).\footnote{Subgaussian random variables are generalizations of Gaussian random variables; their definition can be found in~\cite{Vershynin2011}.}
\begin{cor}[Subgaussian matrices]
	Suppose $\Phi \in \reals^{M \times N}$ is a subgaussian random matrix with independent rows or columns following the construction in~\cite[Thm 5.65]{Vershynin2011}. If
	\begin{eqnarray*}
		M \ge \frac{C_1}{\delta_\manifold^2} \left(D\log\left( \frac{R N}{\tau \delta_\manifold} \right) + \log\left(\frac{V}{\rho}\right)\right) \log\left(\frac{N}{D}\right),
	\end{eqnarray*}
    then with probability greater than $1 - C_2{\rho}$, $\widehat{\Phi} = \Phi D_\xi$ provides a stable embedding of $\manifold$ with conditioning $\delta_\manifold$.
\end{cor} 
The proof of this corollary follows from the fact that such subgaussian random matrices satisfy RIP-($S,\delta$) with high probability whenever $M \ge C_3 {\frac{S}{\delta^2} \log\left(\frac{N}{S}\right)}$~\cite{Baraniuk2008a,Vershynin2011}.
A natural subset of subgaussian random matrices are matrices with i.i.d., symmetric, subgaussian entries of an appropriate subgaussian norm.\footnote{The subgaussian norm of a subgaussian random variable is a generalization of the standard deviation of a Gaussian random variable.}
For this subset of matrices, both $\Phi$ and $\Phi D_\xi$ have the same distribution and thus, the stable embedding for $\manifold$ can actually use just the operator $\Phi$ rather than the operator $\widehat{\Phi}$. 
This last observation formally proves a remark made briefly in~\cite{baraniuk2009random} that stable manifold embeddings can also arise from random subgaussian matrices in addition to random orthoprojectors.

To include a matrix with much more structure (i.e., not having i.i.d.\ entries), we also consider stable manifold embeddings by subsampled Fourier matrices.
\begin{cor}[Subsampled Fourier matrices]
	Suppose $\Phi \in \reals^{M \times N}$ is a subsampled Fourier matrix whose $M$ rows are chosen uniformly at random from the $N \times N$ DFT matrix.\footnote{In fact, this corollary works also for subsampled DTFT matrices~\cite{Rauhut2010b}.}
	If
	\begin{eqnarray*}
		M \ge {\frac{C_1}{\delta_\manifold^2} \left(D\log\left( \frac{R N}{\tau \delta_\manifold} \right) + \log\left(\frac{V}{\rho}\right)\right) \log^4\left(N\right) \log(\rho^{-1})}
	\end{eqnarray*}
	then with probability greater than $1 - C_2 {\rho}$, $\widehat{\Phi} = \Phi D_\xi$ stably embeds $\manifold$ with conditioning $\delta_\manifold$.
\end{cor}
The proof of this corollary comes from the fact that subsampled Fourier matrices satisfy RIP-($S,\delta$) with probability greater than $1 - \rho$ whenever $M \ge C_3 {\frac{S}{\delta^2} \log^4(N) \log(\rho^{-1})}$~\cite{Rudelson2008a,Rauhut2010b}.
For dimensionality reduction problems where the data lies on a manifold, this result provides an efficient measurement scheme whereby the data is  pre-multiplied by a Rademacher sequence and then $M$ coefficients from the  Fourier transform of the data are randomly chosen.

In a similar direction, we also consider stable manifold embeddings from random convolutions.
\begin{cor}[Partial circulant matrices]
	Suppose $\Phi \in \reals^{M \times N}$ is a partial circulant matrix whose first row is made up of i.i.d.\ subgaussian random variables (see~\cite{Krahmer2012} for a detailed construction). If $N$ is large enough and
	\begin{eqnarray*}
		M \ge {\frac{C_1}{\delta_\manifold^2} \left(D \log\left( \frac{R N}{\tau \delta_\manifold} \right) + \log\left(\frac{V}{\rho}\right)\right) \log^4(N)},
	\end{eqnarray*}
	then with probability greater than $1 - C_2{\rho}$, $\widehat{\Phi} = \Phi D_\xi$ stably embeds $\manifold$ with conditioning $\delta_\manifold$.
\end{cor}
Here, the proof follows from the fact that partial circulant matrices satisfy RIP-($S,\delta$) with probability greater than $1 - {N^{-(\log N)(\log^2 S)}}$ (hence the requirement for $N$ to be large enough) whenever $M \ge C_3{\frac{S}{\delta^2} \log^{4}(N)}$ (for $N \ge S$)~\cite{Krahmer2012}.
This again affords us an efficient implementation of a dimensionality reduction scheme for data residing on or near a manifold.
One would first pre-process the data by multiplying its entries with a pre-chosen random Rademacher sequence.
Then, one would simply convolve the processed data with a separate random subgaussian sequence and arbitrarily select $M$ samples of the convolution output.

Before continuing, we note that recent work~\cite{Nelson2012} has shown that by introducing an appropriate (sparse) hashing matrix $H$ on the left of either the partial Fourier or partial circulant matrix $\Phi$ (whose dimensions are also appropriately chosen), the required number of measurements can be reduced by a factor of $\log N$.  Consequently, an application of this hashing matrix (see~\cite{Nelson2012}) as a post-processing step for the two corollaries above can further reduce the dimensions of the measurement space without sacrificing the conditioning of the  stable embedding.

In some situations, one may need to apply the convolution directly on the manifold-modeled data instead of using a pre-processing step (i.e., first multiplying by a diagonal Rademacher matrix). 
For this, consider the matrix $\widehat{\Phi} := R_{\Omega} F D_\xi F^H$ where $F \in \comps^{N \times N}$ is the DFT basis and $R_{\Omega} \in \reals^{M \times N}$ is a restriction operator that selects $M$ entries of a length-$N$ vector (or selects $M$ rows from an $N \times N$ matrix).
Now, this matrix follows our stable embedding construction as $\Phi := R_\Omega F$ is a subsampled Fourier matrix that satisfies the RIP (as long as $M$ is large enough) and $\Psi := F^H$ is orthonormal.
Conveniently, $F D_\xi F^H$ is a circular convolution matrix with $D_\xi$ being the (normalized) Fourier transform of the probe sequence of the convolution.
Thus, the matrix $\widehat{\Phi}$ represents a subsampled convolution operation that can be used to stably embed manifold-modeled data.
This idea is formalized in the follow corollary.
\begin{cor}[Random convolution matrices]
	\label{cor:Rand_conv2}
	Let $C_\xi \in \comps^{N \times N}$ be a random circulant matrix such that $C_\xi := F D_\xi F^H$ where $D_\xi$ is a random diagonal Rademacher matrix and $F$ is the DFT basis.
	Let $\Omega \subset \{1, 2, \cdots, N\}$ with $|\Omega| = M$ be a subset selected uniformly at random.
	If
	\begin{eqnarray*}
		M \ge {\frac{C_1}{\delta_\manifold^2} \left( D \log\left( \frac{RN}{\tau \delta_\manifold} \right) + \log\left( \frac{V}{\rho} \right) \right) \log^4(N)\log(\rho^{-1})},
	\end{eqnarray*}
	then with probability greater than $1 - C_2{\rho}$, $\widehat{\Phi} := R_\Omega C_\xi$ stably embeds $\manifold$ with conditioning $\delta_\manifold$.
\end{cor}
The proof for this corollary follows quickly from the fact that subsampled DFT matrices satisfy RIP-($S,\delta$) with high probability whenever $M \ge C_3{\frac{S}{\delta^2} \log^4(N) \log(\rho^{-1})}$~\cite{Rudelson2008a,Rauhut2010b}.

To address the constraint that some systems can only take localized measurements of the signal, we also consider operators represented by a \emph{Distinct Block Diagonal} (DBD) matrix $\Phi \in \reals^{MJ \times NJ}$ that is non-zero only on the diagonal blocks,
\begin{eqnarray*}
	\Phi = \left(
	\begin{array}{ccc}
		\Phi_1 & & \\
		& \ddots & \\
		& & \Phi_J
	\end{array}
	\right).
\end{eqnarray*}
The blocks $\Phi_j \in \reals^{M \times N}$ on the diagonal consist of i.i.d.\ subgaussian random variables (that are also independent across the blocks).
The following corollary establishes how such matrices can be used to stably embed manifold-modeled data.
\begin{cor}[DBD matrices]
	Let $\Phi \in \reals^{MJ \times NJ}$ be a DBD matrix described above, and let $C_\xi \in \comps^{NJ \times NJ}$ be the circulant matrix as described in Corollary~\ref{cor:Rand_conv2}.
	If $NJ$ is large enough and
	\begin{eqnarray*}
		MJ \ge {\frac{C_1}{\delta_\manifold^2} \left( D \log\left( \frac{RNJ}{\tau \delta_\manifold} \right) + \log\left( \frac{V}{\rho} \right) \right) \log^6(NJ)},
	\end{eqnarray*}
	then with probability greater than $1 - {C_2 \rho}$, $\widehat{\Phi} = \Phi C_\xi$ stably embeds $\manifold$ with conditioning $\delta_\manifold$.
\end{cor}
The proof of this corollary follows quickly from the fact that a DBD matrix $\Phi$ satisfies RIP-($S,\delta$) with probability exceeding $1 - C_3 {(NJ)^{-1}}$ (hence the requirement that $NJ$ is large enough) for frequency sparse signals (i.e., $\Phi F$ satisfies RIP) whenever $MJ \ge C_4 {S \log^6(NJ)}$~\cite{Yap2011a}.
This corollary states that if we pre-process the data by convolving it with a random Rademacher probe, then a block-diagonal matrix (having significantly many more zeros than non-zeros) can stably embed a manifold.

As a last example, the following corollary indicates how one might be able to use a deterministic matrix construction to stably embed manifold-modeled data.
\begin{cor}[Deterministic binary matrices]
	Suppose $\Phi \in \{0,1\}^{M \times N}$ is a deterministic matrix following the construction given in~\cite{Devore2007}. If
	\begin{eqnarray*}
		M \ge {\frac{C_1}{\delta_\manifold^2} \left(D \log\left( \frac{R N}{\tau \delta_\manifold} \right) + \log\left(\frac{V}{\rho}\right)\right)^2 \log^2(N)},
	\end{eqnarray*}
	then with probability greater than $1 - \rho$, $\widehat{\Phi} = \Phi D_\xi$ provides a stable embedding of $\manifold$ with conditioning $\delta_\manifold$.
\end{cor}
Again, this corollary follows from the fact~\cite{Devore2007} that such matrices satisfy RIP-($S,\delta$) whenever $M \ge C_2{\frac{S^2}{\delta^2} \log^2(N)}$.
Despite the additional number of required measurements, deterministic matrices can be of interest to the CS community as it is an NP-hard problem to verify whether a randomly constructed matrix satisfies the RIP~\cite{Tillmann2012}.

\section{Discussions}

In this paper, we showed that all measurement operators $\Phi$ satisfying the RIP can be used to obtain a stable embedding of a manifold.
Moreover, we used this main result to demonstrate several specific examples of stable manifold embeddings that represent efficient dimensionality reduction schemes and operators that model constraints on the measurement process.
These include subsampled Fourier matrices, random convolution matrices, block diagonal matrices, and deterministically constructed matrices.
For each of these operators, we also provided the requisite number of measurements sufficient to ensure a stable embedding of the manifold with high probability and with a predetermined conditioning.
This result represents a combination of two directions of recent interest in the CS community: structured measurement matrices and the development of low-dimensional signal models beyond the canonical sparsity model.

While our main theorem provides a general way to construct manifold embeddings from RIP operators by paying reasonable penalties in the number of measurements, there is room for this result to be improved.  Specifically, Theorem~\ref{thm:main} could be strengthened by removing the logarithmic dependence on the ambient dimension $N$ from the required RIP order $S$.  This reduction by a factor of $\log(N)$ would come at the cost of the proof requiring much more sophisticated machinery involving chaining arguments as described in~\cite[Lemma 3.1]{Clarkson2008a}.\footnote{The fundamental technical consideration is that the current proof technique would have to be extended to consider coverings of the manifold at all scales instead of just a single scale.}  We have chosen to present the current result using a simpler proof technique because even with the improvement described above, the final result would still require a number of measurements that depends on $\log(N)$ due to this factor arising in the RIP requirements for known matrices (as demonstrated in the corollaries of Section~\ref{sec:manStruct}).  Therefore, while this more complex proof technique could reduce the  dependence on $\log(N)$, it could not entirely remove this dependence on the ambient dimension.

\bibliographystyle{IEEEtran}
\bibliography{IEEEabrv,ref}

\begin{thebibliography}{10}
\providecommand{\url}[1]{#1}
\csname url@samestyle\endcsname
\providecommand{\newblock}{\relax}
\providecommand{\bibinfo}[2]{#2}
\providecommand{\BIBentrySTDinterwordspacing}{\spaceskip=0pt\relax}
\providecommand{\BIBentryALTinterwordstretchfactor}{4}
\providecommand{\BIBentryALTinterwordspacing}{\spaceskip=\fontdimen2\font plus
\BIBentryALTinterwordstretchfactor\fontdimen3\font minus
  \fontdimen4\font\relax}
\providecommand{\BIBforeignlanguage}[2]{{%
\expandafter\ifx\csname l@#1\endcsname\relax
\typeout{** WARNING: IEEEtran.bst: No hyphenation pattern has been}%
\typeout{** loaded for the language `#1'. Using the pattern for}%
\typeout{** the default language instead.}%
\else
\language=\csname l@#1\endcsname
\fi
#2}}
\providecommand{\BIBdecl}{\relax}
\BIBdecl

\bibitem{Yap2011d}
H.~L. Yap, M.~B. Wakin, and C.~J. Rozell, ``{Stable Manifold Embeddings with
  Operators Satisfying the Restricted Isometry Property},'' in \emph{Proc.
  Conf. Information Sciences and Systems (CISS)}, Baltimore, MD, 2011.

\bibitem{Baraniuk2008a}
R.~G. Baraniuk, M.~A. Davenport, R.~A. DeVore, and M.~B. Wakin, ``{A Simple
  Proof of the Restricted Isometry Property for Random Matrices},''
  \emph{Constructive Approximation}, vol.~28, no.~3, pp. 253--263, Jan. 2008.

\bibitem{DonohoGrimesISOMAP}
D.~L. Donoho and C.~Grimes, ``{Image Manifolds which are Isometric to Euclidean
  Space},'' \emph{J. Math. Imaging Computer Vision}, vol.~23, no.~1, pp. 5--24,
  Jul. 2005.

\bibitem{Eigenfaces}
M.~Turk and A.~Pentland, ``{Eigenfaces for Recognition},'' \emph{J. Cognitive
  Neuroscience}, vol.~3, no.~1, pp. 71--83, 1991.

\bibitem{Baraniuk2011}
R.~G. Baraniuk, ``{More is Less: Signal Processing and the Data Deluge},''
  \emph{Science}, vol. 331, no. 6018, pp. 717--719, Feb. 2011.

\bibitem{Candes2006}
E.~J. Cand\`{e}s, ``{Compressive Sampling},'' in \emph{Proc. Int. Congress of
  Mathematicians}, vol.~3, 2006, pp. 1433--1452.

\bibitem{Recht2007}
B.~Recht, M.~Fazel, and P.~A. Parrilo, ``{Guaranteed Minimum-Rank Solutions of
  Linear Matrix Equations via Nuclear Norm Minimization},'' \emph{SIAM Review},
  vol.~52, no.~3, pp. 471--501, Jun. 2010.

\bibitem{baraniuk2009random}
R.~G. Baraniuk and M.~B. Wakin, ``{Random Projections of Smooth Manifolds},''
  \emph{Foundations of Computational Mathematics}, vol.~9, no.~1, pp. 51--77,
  2009.

\bibitem{Clarkson2008a}
K.~L. Clarkson, ``{Tighter Bounds for Random Projections of Manifolds},'' in
  \emph{Proc. 24th Annu. Symp. Computational Geometry}.\hskip 1em plus 0.5em
  minus 0.4em\relax ACM, 2008, pp. 39--48.

\bibitem{Shah2011}
P.~Shah and V.~Chandrasekaran, ``{Iterative Projections for Signal
  Identification on Manifolds: Global Recovery Guarantees},'' in \emph{Proc.
  49th Annu. Allerton Conf.}\hskip 1em plus 0.5em minus 0.4em\relax Allerton,
  CA: IEEE, Sep. 2011, pp. 760--767.

\bibitem{Davenport2010a}
M.~A. Davenport, P.~T. Boufounos, M.~B. Wakin, and R.~G. Baraniuk, ``{Signal
  Processing with Compressive Measurements},'' \emph{IEEE J. Selected Topics in
  Signal Processing}, vol.~4, no.~2, pp. 445--460, Apr. 2010.

\bibitem{ISOMAP}
J.~B. Tenenbaum, V.~de~Silva, and J.~C. Langford, ``{A Global Geometric
  Framework for Nonlinear Dimensionality Reduction},'' \emph{Science}, vol.
  290, no. 5500, pp. 2319--2323, Dec. 2000.

\bibitem{hegde2007random}
C.~Hegde, M.~B. Wakin, and R.~G. Baraniuk, ``{Random Projections for Manifold
  Learning: Proofs and Analysis},'' \emph{Technical Report TREE 0710, Rice
  University}, 2007.

\bibitem{Haupt2010}
J.~D. Haupt, W.~U. Bajwa, G.~M. Raz, and R.~D. Nowak, ``{Toeplitz Compressed
  Sensing Matrices with Applications to Sparse Channel Estimation},''
  \emph{IEEE Trans. Information Theory}, vol.~56, no.~11, pp. 5862--5875, 2010.

\bibitem{Rauhut2010c}
H.~Rauhut, J.~K. Romberg, and J.~A. Tropp, ``{Restricted Isometries for Partial
  Random Circulant Matrices},'' \emph{Applied and Computational Harmonic
  Analysis}, vol.~32, pp. 242--254, Oct. 2012.

\bibitem{Krahmer2012}
F.~Krahmer, S.~Mendelson, and H.~Rauhut, ``{Suprema of Chaos Processes and the
  Restricted Isometry Property},'' \emph{arXiv preprint 1207.0235}, Jul. 2012.

\bibitem{Yap2011a}
H.~L. Yap, A.~Eftekhari, M.~B. Wakin, and C.~J. Rozell, ``{The Restricted
  Isometry Property for Block Diagonal Matrices},'' in \emph{Proc. Conf.
  Information Sciences and Systems (CISS)}, Baltimore, MD, 2011.

\bibitem{Devore2007}
R.~A. DeVore, ``{Deterministic Constructions of Compressed Sensing Matrices},''
  \emph{J. Complexity}, vol.~23, no. 4-6, pp. 918--925, Aug. 2007.

\bibitem{Krahmer2010a}
F.~Krahmer and R.~Ward, ``{New and Improved Johnson-Lindenstrauss Embeddings
  via the Restricted Isometry Property},'' \emph{SIAM J. Mathematical
  Analysis}, vol.~43, no.~3, pp. 1269--1281, Sep. 2011.

\bibitem{Johnson84extensionsof}
W.~B. Johnson and J.~Lindenstrauss, ``{Extensions of Lipschitz Mappings into a
  Hilbert Space},'' in \emph{Proc. Conf. Modern Analysis and Probability},
  vol.~26, New Haven, CT, 1984, pp. 189--206.

\bibitem{dasgupta2002elementary}
S.~Dasgupta and A.~Gupta, ``{An Elementary Proof of a Theorem of Johnson and
  Lindenstrauss},'' \emph{Random Structures and Algorithms}, vol.~22, no.~1,
  pp. 60--65, 2002.

\bibitem{Hinton1997}
G.~E. Hinton, ``{Modelling the Manifolds of Images of Handwritten Digits},''
  \emph{IEEE Trans. Neural Networks}, vol.~8, no.~1, pp. 65--74, 1997.

\bibitem{Wakin}
M.~B. Wakin, ``{The Geometry of Low-dimensional Signal Models},'' Ph.D.
  dissertation, Rice University, 2006.

\bibitem{Niyogi2009}
P.~Niyogi, S.~Smale, and S.~Weinberger, ``{Finding the Homology of Submanifolds
  with High Confidence from Random Samples},'' \emph{Technical Report No.
  TR-2004-08, University of Chicago}, pp. 1--23, Mar. 2006.

\bibitem{Gray1974}
A.~Gray, ``{The Volume of a Small Geodesic Ball of a Riemannian Manifold},''
  \emph{The Michigan Mathematical Journal}, vol.~20, no.~4, pp. 329--344, 1974.

\bibitem{Yap2013}
H.~L. Yap, M.~B. Wakin, and C.~J. Rozell, ``{Some Geometric Properties of
  Sampled Sinusoids},'' \emph{in preparation}, 2013.

\bibitem{Vershynin2011}
R.~Vershynin, ``{Introduction to The Non-asymptotic Analysis of Random
  Matrices},'' in \emph{Compressed Sensing, Theory and Applications}, Y.~Eldar
  and G.~Kutyniok, Eds.\hskip 1em plus 0.5em minus 0.4em\relax Cambridge Univ.
  Pr., Nov. 2012, ch.~5, pp. 210--268.

\bibitem{Rauhut2010b}
H.~Rauhut, ``{Compressive Sensing and Structured Random Matrices},'' in
  \emph{Theoretical Foundation and Numerical Methods for Sparse Recovery},
  2010.

\bibitem{Rudelson2008a}
M.~Rudelson and R.~Vershynin, ``{On Sparse Reconstruction from Fourier and
  Gaussian Measurements},'' \emph{Communications in Pure and Applied
  Mathematics}, vol.~61, no.~8, pp. 1025--1045, Aug. 2008.

\bibitem{Nelson2012}
J.~Nelson, E.~Price, and M.~Wootters, ``{New Constructions of RIP Matrices with
  Fast Multiplication and Fewer Rows},'' \emph{arXiv preprint 1211.0986v1},
  Dec. 2012.

\bibitem{Tillmann2012}
A.~M. Tillmann and M.~E. Pfetsch, ``{The Computational Complexity of the
  Restricted Isometry Property, the Nullspace Property, and Related Concepts in
  Compressed Sensing},'' \emph{arXiv preprint 1205.2081}, May 2012.

\bibitem{do1992riemannian}
M.~P. {Do Carmo}, \emph{{Riemannian Geometry}}.\hskip 1em plus 0.5em minus
  0.4em\relax Birkhauser, 1992.

\end{thebibliography}

\appendices

\section{The Condition Number of a Submanifold}
\label{subsec:consequence_cond_num}

The following lemma lists some implications of the condition number for certain geometric properties of the manifold that will be useful to our analysis.
\begin{lemma}
	\label{lem:cond_properties}
	Suppose a submanifold $\manifold \subset \reals^N$ has condition number $\frac{1}{\tau}$.
	Let $p,q \in \manifold$ be two distinct points.
	Then, we have the following three properties of the manifold.
	\begin{enumerate}
		\item (Curvature) If $\gamma(t)$ denotes a unit speed parameterization of the geodesic path joining $p$ and $q$, then $\|\gamma''(t)\|_2 \le \frac{1}{\tau}$.
		Moreover, denoting $\mu := d_\manifold(p,q)$, we have $q - p = \gamma(\mu) - \gamma(0) = \mu \gamma'(0) + r$ with $\|r\|_2 \le \frac{\mu^2}{2\tau}$.
	\item (Twisting of Tangent Spaces) Suppose $d_\manifold(p,q) \le \tau$. Pick $u \in \Tan{p}$ and let $v \in \Tan{q}$ be the parallel transport\footnote{Suppose $\gamma(t)$ denotes a unit speed parameterization of the geodesic path joining $p$ and $q$. By parallel transport~\cite{do1992riemannian}, we mean a vector field $v(t)$ defined along $\gamma(t)$ such that $v(0) = u$, $v(\mu) = v$, $\|v(t)\|_2 = \|u\|_2$, and $\langle v(t), \gamma'(t) \rangle = \langle u, \gamma'(0) \rangle$, where the last two conditions mean that $v(t)$ maintains a constant length and angle with respect to the path $\gamma(t)$.} of $u$ into $\Tan{q}$. Then, $\langle u,v \rangle \ge 1 - \frac{d_\manifold(p,q)}{\tau}$ is guaranteed to hold.\footnote{
Recall that we define the tangent spaces $\Tan{p}, \Tan{q}$ as $D$-dimensional subspaces of $\reals^N$ passing through the origin.
Therefore, we can take the inner product between tangent vectors in different tangent spaces and the inner product is simply the canonical inner product in $\reals^N$.
}
		\item (Self-avoidance) Suppose $\|p - q\|_2 \le \frac{3\tau}{8}$. Then, $\|p - q\|_2 \ge d_\manifold(p,q) - \frac{d_\manifold(p,q)^2}{2\tau}$. As a corollary, we also have $d_\manifold(p,q) \le \tau - \tau \sqrt{1 - \frac{2\|p-q\|_2}{\tau}}$.
	\end{enumerate}
\end{lemma}

The proofs of these properties are simple consequences of the condition number and follow the proofs of similar propositions in~\cite[Section 6]{Niyogi2009}.\footnote{
We note that the proof of Proposition 6.3 in~\cite{Niyogi2009} corresponding to Part 3 of Lemma~\ref{lem:cond_properties} is missing some important details. We have been unable to verify the result stated in~\cite{Niyogi2009} for the range of Euclidean distances ($\|p - q\|_2 \le \tau/2$) stated therein. However, we are able to verify that the result holds for $\|p - q\|_2 \le 3\tau/8$, and we restrict $\|p - q\|_2$ to this range in Part 3 of Lemma~\ref{lem:cond_properties}.
} 
The first property says that the worst case curvature of any unit speed geodesic path along the manifold is bounded by $\frac{1}{\tau}$.
The second property states that for small geodesic distances, the tangent spaces do not ``twist'' too much from one another.
Thus, if we compare a tangent vector to its parallel counterpart in another nearby tangent space, the angle between them is small. 
The last property states that for points on the submanifold close together in Euclidean space, their geodesic and Euclidean distances do not differ much.
Negating the statement, we see that two points with large geodesic distance cannot be arbitrarily close in Euclidean space.

\section{Covering Number of a Manifold}
\label{sec:covering number}

The geodesic regularity $R$ of a manifold $\manifold$ allows us to quantify the geodesic covering number of the manifold (i.e., how many geodesic balls of a certain radius are needed to cover the whole manifold).
More concretely, we say that a  set $\mathcal{C}$ is an $(\epsilon,d_\manifold)$-cover for $\manifold$ if $\manifold \subset \bigcup_{b \in \mathcal{C}} B_\manifold(b,\epsilon)$ where we recall that $B_\manifold(b,\epsilon)$ is the geodesic ball of radius $\epsilon$ centered at $b$.
This implies that for every $x \in \manifold$, we can find a $b \in \mathcal{C}$ such that $d_\manifold(b,x) \le \epsilon$.
The $(\epsilon,d_\manifold)$-cover $\mathcal{C}$ with the minimal cardinality is denoted by $\cover{\manifold}{d_\manifold}{\epsilon}$ and the cardinality of $\cover{\manifold}{d_\manifold}{\epsilon}$ is called the $(\epsilon,d_\manifold)$-covering number of $\manifold$ or simply the geodesic covering number.
The following lemma gives an upper bound on the geodesic covering number of a manifold.
\begin{lemma}
	\label{lem:covering_number}
	The $(\epsilon,d_\manifold)$-covering number of a $D$-dimensional Riemannian submanifold $\manifold \subset \reals^N$ is bounded by
	\begin{eqnarray*}
		\left| \cover{\manifold}{d_\manifold}{\epsilon} \right| \le \frac{V}{\inf_{x \in \manifold} \vol(B_\manifold(x,\frac{\epsilon}{2}))},
	\end{eqnarray*}
	where $V := \vol(\manifold)$.
	If $\manifold$ has geodesic regularity $R$, then
	\begin{eqnarray}
		\left| \cover{\manifold}{d_\manifold}{\epsilon} \right| \le \frac{\left(\frac{2R}{\sqrt{\pi}} \right)^D \left(\sqrt{D/2 + 1} \right)^D V}{\epsilon^D}.
		\label{eq:wakin_regularity}
	\end{eqnarray}
\end{lemma}
The proof of this lemma follows the arguments of the proof of~\cite[Proposition 10.1]{Rauhut2010b}.
We remark that the definition of an equivalent geodesic covering regularity in~\cite{baraniuk2009random} corresponds to $\frac{2R}{\sqrt{\pi}}$ appearing in~\eqref{eq:wakin_regularity}.

We will also need to cover subsets of $\reals^D$ with Euclidean balls (instead of geodesic balls as in the previous lemma).
Thus, we say that the finite set $\cover{\manifold}{\| \cdot \|_2}{\epsilon}$ (of minimal cardinality) is an $(\epsilon,\| \cdot \|_2)$-cover for a subset $\manifold \subset \reals^D$ if $\manifold \subset \bigcup_{b \in \cover{\manifold}{\| \cdot \|_2}{\epsilon}} B_{\reals^D}(b,\epsilon)$. 

\section{Proof of Theorem~\ref{thm:main}}
\label{sec:proof_main}

Mathematically, if we introduce some particular notation, the stable embedding statement~\eqref{eq:SME_def1} can be presented in an equivalent way that is more useful for the desired proof.  First, define the operator $U:\reals^N \setminus \{0\} \rightarrow \sphere^{N-1}$ that takes a non-zero vector and projects it onto the unit sphere (i.e., for any $x \in \reals^N \setminus \{0\}$, $U(x) := \frac{x}{\|x\|_2}$).
$U$ can also act on a subset of $\reals^N$ such that if $\manifold \subset \reals^N$, then $U(\manifold) := \left\{ \frac{x}{\|x\|_2} \;|\; x \in \manifold \setminus \{0\} \right\}$.
Next, we define the difference between any two subsets $A - B$ (with $A,B \subset \reals^N$) as the set comprised of pairwise differences between the elements of the sets, $A - B := \left\{ a - b \;|\; a \in A, \; b \in B \right\}$.
Finally, for a finite subset $\manifold$ of $\reals^N$, $|\manifold|$ denotes its cardinality.

Suppose $\manifold$ is the Riemannian submanifold considered in Theorem~\ref{thm:main} and define the set of \emph{chords} of $\manifold$ (i.e., the set of all normalized difference vectors in $\manifold$) as
\begin{eqnarray*}
		U(\manifold - \manifold) = \left\{\left.\frac{x - y}{\|x - y\|_2} \;\right|\; x, y \in \manifold , \;x \neq y\right\}.
\end{eqnarray*} 
Then, $\widehat{\Phi}$ provides a stable embedding of $\manifold$ with conditioning $\delta_\manifold$ if and only if $\sup_{x \in U(\manifold - \manifold)} \left| \|\widehat{\Phi} x\|_2^2 - 1 \right| \le \delta_\manifold$. 
In other words, $\widehat{\Phi}$ provides a stable embedding of $\manifold$ if and only if $\widehat{\Phi}$ approximately preserves the norms of all  elements in $U(\manifold - \manifold)$.
This equivalence follows immediately from the stable embedding definition, Definition~\ref{def:stable_emb_def}, after a judicious rearrangement of the variables. 
We will use this equivalence for the proof of Theorem~\ref{thm:main}.

The proof of Theorem~\ref{thm:main} follows very closely the proof technique of~\cite{baraniuk2009random} and is basically comprised of three steps.
The first step involves judiciously choosing a \emph{generalized} covering set $B$ of the manifold $\manifold$ using a collection of points on the manifold and their corresponding tangent planes.
Lemma~\ref{lem:cover} then shows that every point of $U(\manifold-\manifold)$ can be approximated by some point in $U(B-B)$. 
The second step (encapsulated by Lemma~\ref{lem:apply_Ward}) then applies the JL lemma for RIP operators (i.e., Theorem~\ref{thm:Ward}) to obtain an approximate norm preservation of all elements of $U(B-B)$. 
Finally in Section~\ref{subsec:syn}, we extend this approximate norm preservation to all points on $U(\manifold-\manifold)$ via simple geometric arguments. 
As described, the proof technique here distinguishes from that of~\cite{baraniuk2009random} mainly in the separation of the stable embedding operator from the covering of the manifold. 

\subsection{Covering $U(\manifold- \manifold)$}

In this section, we construct a set $B$ and show in Lemma~\ref{lem:cover} that $U(B-B)$ is a cover of $U(\manifold- \manifold)$. 
Let $A = A(T) := \cover{\manifold}{d_\manifold}{T}$ for some $T \le \frac{3 \tau}{4}$ be the $(T, d_\manifold)$-cover of $\manifold$ of minimum cardinality.
For any $x \in \manifold$, we can find an $a \in A$ such that $d_\manifold(a,x) \le T$. 
Define a generalized covering set $B$ of the manifold $\manifold$ as
\begin{eqnarray*}
	B = B(T) = \bigcup_{a \in A} \{a + \Tan{a}(T)\},
\end{eqnarray*}
where $\Tan{a}(T) := \{ u \in \Tan{a} \;|\; \|u\|_2 \le T \}$ refers to all tangent vectors of $\manifold$ at the point $a$ whose lengths are less than $T$.\footnote{
Recall that we had defined the tangent space $\Tan{a}$ as a $D$-dimensional subspace of $\reals^N$ passing through the origin. Therefore when considering the Euclidean length of $u \in \Tan{a}$ and adding $u$ to the point $a \in \manifold$, we perform these operations in $\reals^N$.
} 
$B$ is called a generalized covering set as it is a union of (subsets of) affine $D$-dimensional planes of $\reals^N$ (i.e., $B$ is not a finite set). 

The goal of this section is to show that $U(B-B)$ is a suitable cover of $U(\manifold- \manifold)$ as detailed in the following lemma:
\begin{lemma}
	\label{lem:cover}
	Let $B$ be defined as above.
	For $T \le \frac{3\tau}{4}$, set $\epsilon(T) := 4\sqrt{\frac{T}{\tau}}$.
	Then, $U(B-B)$ is an $(\epsilon(T), \| \cdot \|_2)$-cover of $U(\manifold- \manifold)$.
	In other words, for every $u \in U(\manifold - \manifold)$, we can find a $b \in U(B-B)$ such that $\|u - b\|_2 \le \epsilon(T)$.
\end{lemma}

\begin{proof} 
To prove this lemma, we break the set of chords $U(\manifold- \manifold)$ into  sets of ``long'' and ``short'' chords which we will cover separately.
The sets of short and long chords (delineated by Euclidean distance $\frac{T}{2}$) are defined by:
\begin{eqnarray*}
	U^s(\manifold - \manifold) &:=& \left\{ \frac{x_1 - x_2}{\|x_1 - x_2\|_2} \;|\; x_1, x_2 \in \manifold, 0 < \|x_1 - x_2\|_2 \le \frac{T}{2} \right\}, \\
	U^l(\manifold - \manifold) &:=& \left\{ \frac{x_1 - x_2}{\|x_1 - x_2\|_2} \;|\; x_1, x_2 \in \manifold, \|x_1 - x_2\|_2 > \frac{T}{2} \right\},
\end{eqnarray*}
and $U(\manifold- \manifold) = U^s(\manifold - \manifold) \cup U^l(\manifold - \manifold)$.

Let us start with the cover of $U^s(\manifold - \manifold)$ where, due to the locally Euclidean structure of manifolds, the short chords in $U^s(\manifold - \manifold)$ can be approximated by tangent vectors of the manifold.
Pick an element $\frac{x_1 - x_2}{\|x_1 - x_2\|_2}$ of $U^s(\manifold - \manifold)$ where by definition $\|x_1 - x_2\|_2 \le \frac{T}{2}$.
From Lemma~\ref{lem:cond_properties}, $\|x_1 - x_2\|_2 \le \frac{T}{2} \le \frac{3\tau}{8}$ (since we assume $T \le \frac{3\tau}{4}$) implies that
\begin{eqnarray}
	d_\manifold(x_1,x_2) \le \tau - \tau\sqrt{1 - \frac{2\|x_1 - x_2\|_2}{\tau}} \le \tau - \tau \left(1 - \frac{2\|x_1 - x_2\|_2}{\tau} \right) = 2\|x_1 - x_2\|_2 \le T.
	\label{eq:bound_on_mu}
\end{eqnarray}
Now, let $\mu := d_\manifold(x_1, x_2)$ and let $\gamma(t)$ be the unit-speed geodesic parameterization from $x_1$ to $x_2$ where $\gamma(0) = x_1$, $\gamma(\mu) = x_2$, and $\gamma'(0) \in U(\Tan{x_1})$.
From Lemma~\ref{lem:cond_properties}, we have
\begin{eqnarray}
	x_1 - x_2 = \gamma(\mu) - \gamma(0) = \mu \gamma'(0) + r,
	\label{eq:short_chord_para}
\end{eqnarray}
with $\|r\|_2 \le \frac{\mu^2}{2\tau}$. 
Let $a \in A$ be the closest geodesic covering point to $x_1$ (so that $d_\manifold(a,x_1) \le T$) and let $b \in U(\Tan{a})$ be the parallel transport of $\gamma'(0)$ onto $\Tan{a}$. 
First, $b \in U(B-B)$ by definition of the set $B$ and second, Lemma~\ref{lem:cond_properties} says that $\langle \gamma'(0), b \rangle \ge 1 - \frac{d_\manifold(a,x_1)}{\tau} \ge 1 - \frac{T}{\tau}$, since $d_\manifold(a,x_1) \le T \le \frac{3\tau}{4} \le \tau$. 
Thus,
\begin{eqnarray}
	\|\gamma'(0) - b\|_2^2 = \|\gamma'(0)\|_2^2 + \|b\|_2^2 - 2 \langle \gamma'(0), b \rangle = 2(1 - \langle \gamma'(0), b \rangle) \le  \frac{2T}{\tau}.
	\label{eq:bound_gammab}
\end{eqnarray}
We now show that $b$ is indeed close to the short chord $\frac{x_1 - x_2}{\|x_1 - x_2\|_2}$ by combining~\eqref{eq:short_chord_para} and~\eqref{eq:bound_gammab}:
\begin{eqnarray}
	\left\|\frac{x_1 - x_2}{\|x_1 - x_2\|_2} - b \right\|_2 &=& \left\|\frac{\mu \gamma'(0) + r}{\|x_1 - x_2\|_2} - b \right\|_2 \nonumber \\
	&=& \left\|\gamma'(0) - b + \left(\frac{\mu}{\|x_1 - x_2\|_2} - 1\right)\gamma'(0) + \frac{r}{\|x_1 - x_2\|_2} \right\|_2 \nonumber \\
	&\le& \sqrt{\frac{2T}{\tau}} + \left(\frac{\mu}{\|x_1 - x_2\|_2} - 1\right) + \frac{\mu}{\|x_1 - x_2\|_2} \cdot \frac{\mu}{2 \tau}. 
	\label{eq:short_chord_1} 
\end{eqnarray}
To remove the dependence of~\eqref{eq:short_chord_1} on $\|x_1 - x_2\|_2$, we use Lemma~\ref{lem:cond_properties} to obtain
\begin{eqnarray*}
	\|x_1 - x_2\|_2 \ge \mu - \frac{\mu^2}{2\tau} = \mu \left(1 - \frac{\mu}{2\tau} \right) \Leftrightarrow \frac{\mu}{\|x_1 - x_2\|_2} \le  \frac{1}{\left(1 - \frac{\mu}{2\tau} \right)} \le 1 + \frac{\mu}{\tau},
\end{eqnarray*}
where we used the inequality $\frac{1}{1 - a} \le (1 + 2a)$ whenever $0 < a \le \frac{1}{2}$ (this is true since $\frac{\mu}{2\tau} \le \frac{T}{2\tau} \le \frac{3}{8}$). 
Applying this to~\eqref{eq:short_chord_1} and applying $\mu \le T$ obtained earlier in~\eqref{eq:bound_on_mu}, we obtain
\begin{eqnarray*}
	\left\|\frac{x_1 - x_2}{\|x_1 - x_2\|_2} - b \right\|_2 \le
	\sqrt{\frac{2T}{\tau}} + \frac{\mu}{\tau} + \left( 1 + \frac{\mu}{\tau} \right)\frac{\mu}{2\tau} \le 4 \sqrt{\frac{T}{\tau}} =: \epsilon_1(T), 
\end{eqnarray*}
where we used the fact that $\frac{T^2}{\tau^2} \le \frac{T}{\tau} \le \sqrt{\frac{T}{\tau}} \le 1$. 
This proves that for every element of $U^s(\manifold - \manifold)$, we can find an element $b \in U(B-B)$ that is within $\epsilon_1(T)$ of it.
Thus, $U(B-B)$ is an $(\epsilon_1(T), \| \cdot \|_2)$-cover of $U^s(\manifold - \manifold)$.

Let us now move on to covering $U^l(\manifold - \manifold)$. 
Pick an element $\frac{x_1 - x_2}{\|x_1 - x_2\|_2}$ of $U^l(\manifold - \manifold)$. 
For each $x_i$, for $i = 1, 2$, choose its closest geodesic covering point $a_i \in A$ so that $\mu_i := d_\manifold(a_i, x_i) \le T$.
Let $\gamma_i(t)$ be the unit-speed geodesic parameterization from $a_i$ to $x_i$, so that $\gamma_i(0) = a_i$, $\gamma_i(\mu_i) = x_i$, and $\gamma_i'(0) \in U(\Tan{a_i})$. 
From Lemma~\ref{lem:cond_properties}, we have $x_i - a_i = \gamma_i(\mu_i) - \gamma_i(0) = \mu_i \gamma_i'(0) + r_i$, with $\|r_i\|_2 \le \frac{\mu_i^2}{2\tau}$.
Define $b_i = a_i + \mu_i \gamma_i'(0)$ where it is clear that $x_i - b_i = r_i$ and $b_i \in \{a_i + \Tan{a_i}(T)\} \subset B$.
We will use $\frac{b_1 - b_2}{\|b_1 - b_2\|_2} \in U(B-B)$ as a covering point near to $\frac{x_1 - x_2}{\|x_1 - x_2\|_2}$. 
Following~\cite{Clarkson2008a}, we have
\begin{eqnarray*}
	\left\| \frac{x_1 - x_2}{\|x_1 - x_2\|_2} - \frac{b_1 - b_2}{\|b_1 - b_2\|_2} \right\|_2
	&=& \left\| \frac{(x_1 - x_2) - (b_1 - b_2)}{\|x_1 - x_2\|_2} + \frac{(b_1 - b_2)(\|b_1 - b_2\|_2 - \|x_1 - x_2\|_2)}{\|x_1 - x_2\|_2 \|b_1 - b_2\|_2} \right\|_2 \\
	&\le& \left\| \frac{(x_1 - x_2) - (b_1 - b_2)}{\|x_1 - x_2\|_2} \right\|_2 + \left\| \frac{(b_1 - b_2)(\|b_1 - b_2\|_2 - \|x_1 - x_2\|_2)}{\|x_1 - x_2\|_2 \|b_1 - b_2\|_2} \right\|_2.
\end{eqnarray*}
We will calculate each of the terms separately.
For the first term, we see that
\begin{eqnarray*}
	\left\| \frac{(x_1 - x_2) - (b_1 - b_2)}{\|x_1 - x_2\|_2} \right\|_2 &=& \left\| \frac{(x_1 - b_1) - (x_2 - b_2)}{\|x_1 - x_2\|_2} \right\|_2 \\
	&\le& \frac{1}{\|x_1 - x_2\|_2}\left( \frac{\mu_1^2}{2\tau} + \frac{\mu_2^2}{2\tau} \right) \\
	&\le& \frac{T^2}{\tau \|x_1 - x_2\|_2}.
\end{eqnarray*}
For the second term, we have
\begin{eqnarray*}
	\left\| \frac{(b_1 - b_2)(\|b_1 - b_2\|_2 - \|x_1 - x_2\|_2)}{\|x_1 - x_2\|_2 \|b_1 - b_2\|_2} \right\|_2
	&=& \frac{\left|\|x_1 - x_2\|_2 -\|b_1 - b_2\|_2  \right|}{\|x_1 - x_2\|_2} \\
	&\le&  \frac{\|(x_1 - x_2) - (b_1 - b_2)\|_2}{\|x_1 - x_2\|_2} \\
	&=&  \frac{\|(x_1 - b_1) - (x_2 - b_2)\|_2}{\|x_1 - x_2\|_2} \\
	&\le& \frac{T^2}{\tau \|x_1 - x_2\|_2},
\end{eqnarray*}
where we used the reverse triangle inequality in the second line. 
Now, our definition of long chords implies that $\|x_1 -x_2\|_2 > \frac{T}{2}$.
Therefore,
\begin{eqnarray*}
	\left\| \frac{x_1 - x_2}{\|x_1 - x_2\|_2} - \frac{b_1 - b_2}{\|b_1 - b_2\|_2} \right\|_2 &<& \frac{4T}{\tau} =: \epsilon_2(T).
\end{eqnarray*}
Thus, $U(B-B)$ is an $(\epsilon_2(T), \| \cdot \|_2)$-cover of $U^l(\manifold - \manifold)$.

Putting everything together, since $\epsilon_1(T) = 4\sqrt{\frac{T}{\tau}} \ge 4\frac{T}{\tau} = \epsilon_2(T)$, we have that $U(B-B)$ is a $(4\sqrt{\frac{T}{\tau}}, \|\cdot\|_2)$-cover of $U(\manifold-\manifold)$. 
\end{proof}

\subsection{Applying the JL Lemma}

We want to use $U(B-B)$ as a proxy for $U(\manifold-\manifold)$ for applying Theorem~\ref{thm:Ward}.
However, $U(B-B)$ is not just a finite collection of points and thus Theorem~\ref{thm:Ward} cannot be applied directly.
Fortunately, it is well-known that unit spheres on planes (or affine planes) can be well-covered by a finite collection of points.
Indeed, as a corollary to Lemma~\ref{lem:covering_number} (see also~\cite{Rauhut2010b}), the $(\epsilon, \|\cdot\|_2)$-covering number of a $D$-dimensional sphere is $\left(1 + \frac{2}{\epsilon}\right)^D$. 

$U(B-B)$ can be divided into two sets of elements, namely:
\begin{enumerate} 
	\item $B_1 := U\left( \bigcup_{a \in A} \{\Tan{a}(T) - \Tan{a}(T)\} \right) = U\left( \bigcup_{a \in A} \Tan{a} \right)$, and
	\item $B_2 := U\left( \bigcup_{a_1,a_2 \in A, a_1 \neq a_2} \left\{ (a_1 - a_2) + (\Tan{a_1}(T) - \Tan{a_2}(T)) \right\} \right)$. 
\end{enumerate} 
The set $B_1$ is comprised of $|A|$ $D$-dimensional unit spheres.
From our earlier discussion, we know that each unit sphere can be $(\epsilon, \|\cdot\|_2)$-covered by at most $(1 + \frac{2}{\epsilon})^D$ points.
Thus, $|\cover{B_1}{\|\cdot\|_2}{\epsilon}| \le |A|(1 + \frac{2}{\epsilon})^D$. 
The set $B_2$ is the projection onto the unit sphere (in $\reals^N$) of not more than $|A|^2$ subsets of affine planes where each affine plane is contained in a linear subspace of dimension $2D+1$. 
Thus, $|\cover{B_2}{\|\cdot\|_2}{\epsilon}| \le |A|^2(1 + \frac{2}{\epsilon})^{2D+1}$. 

Define the collection of points $E(\epsilon) := \cover{B_1}{\|\cdot\|_2}{\epsilon} \cup \cover{B_2}{\|\cdot\|_2}{\epsilon}$. 
From our previous discussion, the cardinality of $E(\epsilon)$ is bounded by
\begin{eqnarray}
	|E(\epsilon)| \le |A|\left(1 + \frac{2}{\epsilon} \right)^D + |A|^2 \left(1 + \frac{2}{\epsilon} \right)^{2D+1}	
	\le 2 |A|^2 \left(1 + \frac{2}{\epsilon} \right)^{2D+1}.
	\label{eq:size_of_E}
\end{eqnarray}
By construction, for any $b \in U(B-B)$, we can find an $e \in E(\epsilon)$ such that $\|b - e\|_2 \le \epsilon$.
With the aid of $E(\epsilon)$, we can show the stable embedding of $U(B-B)$ by the operator $\widehat{\Phi}$ defined in Theorem~\ref{thm:main}.
\begin{lemma}
	\label{lem:apply_Ward}
	Choose any failure probability $\rho$ and conditioning $\delta_\manifold' \le \frac{4}{9}$.
	Set the covering resolution $\epsilon$ in the set $E(\epsilon)$ to $\epsilon = \frac{\delta_\manifold'}{N+1}$.
	Suppose we have a matrix $\Phi$ satisfying the RIP of order $S \ge 40 \log \left( \frac{4|E(\epsilon)|}{\rho} \right)$ and conditioning $\delta \le \frac{\delta_\manifold'}{4}$.
	Then, with probability exceeding $1 - \rho$, the matrix $\widehat{\Phi} := \Phi D_\xi \Psi$ is a (non-squared) stable embedding\footnote{
	Squared and non-squared stable embeddings differ only by a small constant in their conditioning.
	To be more concrete, suppose $C \subset \reals^N$. Then it can be shown that $\sup_{c \in C} \left| \|\widehat{\Phi} c\|_2 - 1 \right| \le \sup_{c \in C} \left| \|\widehat{\Phi} c\|_2^2 - 1 \right|$. Furthermore if $\sup_{c \in C} \left| \|\widehat{\Phi} c\|_2 - 1 \right| \le 1$, then it can be shown that $\sup_{c \in C} \left| \|\widehat{\Phi} c\|_2^2 - 1 \right| \le 3 \sup_{c \in C} \left| \|\widehat{\Phi} c\|_2 - 1 \right|$.
	} of $B$ with conditioning $\frac{9}{4} \delta_\manifold'$ (i.e., $\sup_{b \in U(B-B)} \left| \|\widehat{\Phi} b\|_2 - 1 \right| \le  \frac{9}{4} \delta_\manifold'$). 
\end{lemma}

\begin{proof}[Proof of Lemma~\ref{lem:apply_Ward}]
	Fix $\rho < 1$ and $\delta_\manifold' \le \frac{4}{9}$. 
If $\Phi$ satisfies RIP-$(S,\delta)$ with $S \ge 40 \log \left( \frac{4|E(\epsilon)|}{\rho} \right)$ (with $\epsilon$ to be defined later) and $\delta \le \frac{\delta_\manifold'}{4}$ as assumed in Lemma~\ref{lem:apply_Ward}, then Theorem~\ref{thm:Ward} states that with probability exceeding $1 - \rho$, $\sup_{e \in E(\epsilon)} \left| \|\widehat{\Phi} e\|_2 - 1 \right| \le \sup_{e \in E(\epsilon)} \left| \|\widehat{\Phi} e\|_2^2 - 1 \right| \le \delta_\manifold'$. 
For a fixed $b \in U(B-B)$, find its nearest covering point $e \in E(\epsilon)$ such that $\|b - e\|_2 \le \epsilon$. 
Then, we have
\begin{eqnarray}
	\label{eq:upperbound_b}
	\|\widehat{\Phi} b\|_2 \le \|\widehat{\Phi} e\|_2 + \left\|\widehat{\Phi} (b - e) \right\|_2
	\le (1 + \delta_\manifold')+ \left\|\Phi\right\|_2 \left\| D_\xi \Psi (b - e)\right\|_2.
\end{eqnarray}
Now, it is easy to show that for a matrix $\Phi \in \comps^{M \times N}$ that satisfies RIP-($S,\delta$), $\|\Phi\|_2 \le \left( \frac{N}{S} + 1 \right)( 1+ \delta)$. 
Applying this fact to~\eqref{eq:upperbound_b}, we have
\begin{eqnarray*}
	\|\widehat{\Phi} b\|_2 \le (1 + \delta_\manifold')+ \left( \frac{N}{S} + 1 \right)( 1+ \delta) \epsilon
	\le (1 + \delta_\manifold')+ \left( N + 1 \right) \left(1+ \frac{\delta_\manifold'}{4}\right) \epsilon.
\end{eqnarray*}
To remove the catastrophic dependence on $(N+1)$, set $\epsilon = \frac{\delta_\manifold'}{N+1}$.
Using this choice of $\epsilon$, we have
\begin{eqnarray*}
	\|\widehat{\Phi} b\|_2
	\le (1 + \delta_\manifold')+ \left(1+ \frac{\delta_\manifold'}{4}\right) \delta_\manifold'
	\le 1 + 2 \delta_\manifold' + \frac{(\delta_\manifold')^2}{4}
	\le 1 + \frac{9}{4} \delta_\manifold',
\end{eqnarray*}
where we used the fact that $\delta_\manifold' \le \frac{4}{9} \le  1$.
Using the same steps for the lower bound, we obtain
\begin{eqnarray*}
	\|\widehat{\Phi} b\|_2
	\ge \|\widehat{\Phi} e\|_2 - \left\|\widehat{\Phi} (b - e) \right\|_2
	\ge (1 - \delta_\manifold') - \left( N + 1 \right) \left(1+ \frac{\delta_\manifold'}{4}\right) \epsilon 
	\ge 1 - \frac{9}{4} \delta_\manifold'. 	
\end{eqnarray*}
Since the upper and lower bounds coincide, and they are valid for all $b \in U(B-B)$, we arrive at our required conclusion. 
\end{proof}

\subsection{Synthesis}
\label{subsec:syn}

Finally, it remains to extend the stable embedding from $U(B - B)$ to $U(\manifold-\manifold)$. 
From Lemma~\ref{lem:cover}, for any $u \in U(\manifold-\manifold)$, we can find a $b \in U(B-B)$ such that $\|b - u\|_2 \le \epsilon(T)$ with $\epsilon(T) = 4 \sqrt{\frac{T}{\tau}}$. 
Using Lemma~\ref{lem:apply_Ward} (with $\rho$ fixed and $\delta_\manifold' \le \frac{4}{9}$ to be defined later), triangle inequalities, and the fact that $\|\Phi\|_2 \le \left( \frac{N}{S} + 1 \right)\left( 1+ \frac{\delta_\manifold'}{4}\right)$, we have
\begin{eqnarray}
	\|\widehat{\Phi} u\|_2 \le \|\widehat{\Phi} b\|_2 + \|\Phi\|_2 \|D_\xi \Psi (b - u)\|_2 
\le (1 + \frac{9}{4} \delta_\manifold') + \left(1 + \frac{\delta_\manifold'}{4} \right) \left( N + 1 \right) \epsilon(T).
	\label{eq:wphiu_upper}
\end{eqnarray}
Set $T$ such that $\epsilon(T) = \frac{\delta_\manifold'}{N+1}$.
By using the formula for $\epsilon(T)$, we have that $T = \frac{(\delta_\manifold')^2 \tau}{16(N+1)^2}$.
It is easy to check that $T \le \frac{3\tau}{8}$, which fulfills the condition of Lemma~\ref{lem:cover}.
Plugging this choice of $\epsilon(T)$ into~\eqref{eq:wphiu_upper}, we get
\begin{eqnarray*}
	\|\widehat{\Phi} u\|_2 \le
	(1 + \frac{9}{4} \delta_\manifold') + \left(1 + \frac{\delta_\manifold'}{4} \right) \delta_\manifold' \le 1 + \frac{7}{2}\delta_\manifold',
\end{eqnarray*}
where we used the fact that $\delta_\manifold' \le \frac{4}{9} < 1$.
For the lower conditioning bound, we use the same estimates to arrive at
\begin{eqnarray*}
	\|\widehat{\Phi} u\|_2 \ge \|\widehat{\Phi} b\|_2 - \|\Phi\|_2 \|D_\xi \Psi (b - u)\|_2
	\ge (1 - \frac{9}{4} \delta_\manifold') - \left(1 + \frac{\delta_\manifold'}{4} \right) \delta_\manifold' \ge 1 - \frac{7}{2} \delta_\manifold'.
\end{eqnarray*}
Since the upper and lower bounds coincide, we have via the squared and non-squared conditioning bounds 
\begin{eqnarray*}
	\sup_{u \in U(\manifold-\manifold)} \left| \|\widehat{\Phi} u\|_2^2 -1 \right| \le 3 \sup_{u \in U(\manifold-\manifold)} \left| \|\widehat{\Phi} u\|_2 -1 \right| \le \frac{21}{2}\delta_\manifold'. 
\end{eqnarray*}

It remains to do some bookkeeping.
First, given a predetermined stable manifold embedding conditioning $\delta_\manifold < 1$, set $\delta_\manifold' = \frac{2}{21} \delta_\manifold$.
It is clear that this choice of $\delta_\manifold'$ validates the assumption that $\delta_\manifold' \le \frac{4}{9}$ in Lemma~\ref{lem:apply_Ward}, and we have $\sup_{u \in U(\manifold-\manifold)} \left| \|\widehat{\Phi} u\|_2^2 -1 \right| \le \delta_\manifold$ which is what we are trying to prove.
Next, according to the JL lemma for RIP operators (Lemma~\ref{lem:apply_Ward}), the RIP conditioning for the matrix $\Phi$ needs to satisfy $\delta \le \frac{\delta_\manifold'}{4} = \frac{\delta_\manifold}{42}$.
This is the condition for the RIP conditioning in  Theorem~\ref{thm:main}.
Finally, according to the JL lemma for RIP operators (Lemma~\ref{lem:apply_Ward}), the RIP order needs to satisfy $S \ge 40 \log \left( \frac{4|E(\delta_\manifold'/(N+1))|}{\rho} \right)$.
For this, we need do some work.
First, using~\eqref{eq:size_of_E}, we have $\left|E \left(\frac{\delta_\manifold'}{N+1} \right)\right| = \left|E \left(\frac{2\delta_\manifold }{21(N+1)} \right) \right| \le 2|A|^2 \left(1 + \frac{21(N+1)}{\delta_\manifold} \right)^{2D+1}$. 
Now $|A|$ depends on the geodesic covering resolution $T$, which was set to be $T = \frac{(\delta_\manifold')^2 \tau}{16(N+1)^2} = \frac{\delta_\manifold^2 \tau}{1764 (N+1)^2}$. 
By Lemma~\ref{lem:covering_number}, which gives the geodesic number of a manifold with geodesic regularity $R$, we have
\begin{eqnarray*}
	\log\left(|A| \right) &\le& \log\left(\frac{\left(\frac{2R}{\sqrt{\pi}} \right)^D \left(\sqrt{{D}/{2} + 1} \right)^D V}{T^D}\right)  \\
	&=& \log\left(\frac{\left(\frac{3528 R}{\sqrt{\pi}} \right)^D \left(\sqrt{{D}/{2} + 1} \right)^{D} (N+1)^{2D} V}{\delta_\manifold^{2D} \tau^D} \right) \\
	&=& D \log \left(\frac{3528 R \left(\sqrt{D/2 + 1} \right) (N+1)^2}{\sqrt{\pi} {\delta_\manifold^2} \tau} \right) + \log(V).
\end{eqnarray*}
Putting everything together, the order $S$ of the RIP of the matrix $\Phi$ must satisfy
\begin{eqnarray*}
	S \ge 40 \left( 2D \log \left(\frac{3528 R \left(\sqrt{D/2 + 1} \right) (N+1)^2}{\sqrt{\pi} {\delta_\manifold^2} \tau} \right)  + (2D+1)\log\left(1 + \frac{21(N+1)}{\delta_\manifold} \right) + \log \left(\frac{8 V^2}{\rho}\right) \right). 	
\end{eqnarray*}
This concludes the proof of Theorem~\ref{thm:main}.

\end{document}